\documentclass[conference]{IEEEtran}
\usepackage[left=0.68in,right=0.68in,top=0.60in,bottom=0.65in]{geometry}
\usepackage[T1]{fontenc}
\usepackage[latin9]{inputenc}
\usepackage{amsthm}
\usepackage{amsmath} 
\usepackage{graphicx} 
\usepackage{color} 
\usepackage{cite}
\usepackage{algpseudocode}
\usepackage{algorithm}
\usepackage{amssymb}
\usepackage{subfigure}
\usepackage{esint}
\usepackage{algpseudocode}
\usepackage{epstopdf}
\usepackage{multirow}
\usepackage{makecell}
\usepackage{float}
\usepackage{lipsum}
\usepackage[export]{adjustbox}

\newtheorem{definition}{Definition}
\newtheorem{remark}{Remark}

\newtheorem{example}{Example}

\theoremstyle{plain}
\theoremstyle{plain}
\newtheorem{thm}{Theorem}
\newtheorem{lemma}{Lemma}

\newcommand{\comment}[1]{}

\IEEEoverridecommandlockouts

\begin{document}

\title{\vspace{-0.5em}High Performance Non-Binary Spatially-Coupled Codes for Flash Memories\vspace{-0.4em}}

\author{
   \IEEEauthorblockN{Ahmed Hareedy, Homa Esfahanizadeh, and Lara Dolecek$$}
   \IEEEauthorblockA{$$Electrical Eng. Department, University of California, Los Angeles, Los Angeles, CA 90095 USA \\
                                  \{ahareedy, hesfahanizadeh\}@ucla.edu and dolecek@ee.ucla.edu}\vspace{-2.5em}
}
\maketitle

\begin{abstract}
Modern dense Flash memory devices operate at very low error rates, which require powerful error correcting coding (ECC) techniques. An emerging class of graph-based ECC techniques that has broad applications is the class of spatially-coupled (SC) codes, where a block code is partitioned into components that are then rewired multiple times to construct an SC code. Here, our focus is on SC codes with the underlying circulant-based structure. In this paper, we present a three-stage approach for the design of high performance non-binary SC (NB-SC) codes optimized for practical Flash channels; we aim at minimizing the number of detrimental general absorbing sets of type two (GASTs) in the graph of the designed NB-SC code. In the first stage, we deploy a novel partitioning mechanism, called the \textit{optimal overlap partitioning}, which acts on the protograph of the SC code to produce optimal partitioning corresponding to the smallest number of detrimental objects. In the second stage, we apply a new \textit{circulant power optimizer} to further reduce the number of detrimental GASTs. In the third stage, we use the \textit{weight consistency matrix framework} to manipulate edge weights to eliminate as many as possible of the GASTs that remain in the NB-SC code after the first two stages (that operate on the unlabeled graph of the code). Simulation results reveal that NB-SC codes designed using our approach outperform state-of-the-art NB-SC codes when used over Flash channels.
\end{abstract}

\vspace{-0.24em}
\section{Introduction}\label{sec_intro}
\vspace{-0.14em}

Because of their excellent performance, graph-based codes are among the most attractive error correction techniques deployed in modern storage devices \cite{jia_sym, maeda_fl}. Non-binary (NB) codes offer superior performance over binary codes, and are thus well suited for modern Flash~memories. The nature of~the detrimental objects that dominate the error floor region of non-binary graph-based codes depends on the underlying channel of the device. Unlike in the case of canonical channels, in a recent research \cite{ahh_jsac}, it was revealed that general absorbing sets of type two (GASTs) are the objects that dominate the error floor of NB graph-based codes over practical, inherently asymmetric Flash channels \cite{ahh_jsac, mit_nl}. We analyzed GASTs, and proposed a combinatorial framework, called the weight consistency matrix (WCM) framework, that removes GASTs from the Tanner graph of NB codes, and results in at least $1$ order of magnitude performance gain over asymmetric Flash channels \cite{ahh_jsac, ahh_tit}.

A particular class of graph-based codes that has received recent attention is the class of spatially-coupled (SC) codes \cite{fels_sc}. SC codes are constructed via partitioning an underlying LDPC code into components, and then coupling them together multiple times. Recent results on SC codes include asymptotic analysis, e.g., \cite{kud_sc}, and finite length designs, e.g., \cite{pus_sc, mitch_sc, iye_sc}. Non-binary SC (NB-SC) codes designed using cutting vector (CV) partitioning and optimized for 1-D magnetic recording applications were introduced in \cite{homa_sc}. The idea of partitioning the underlying block code by minimizing the overlap of its rows of circulants (so called minimum overlap (MO)) was recently introduced and applied to AWGN channels in \cite{homa_mo}.

In this paper, we present the first study of NB-SC codes designed for practical Flash channels. The underlying block codes we focus on are circulant-based (CB) codes. Our combinatorial approach to design NB-SC codes comprises three stages. The first two stages aim at optimizing the unlabeled graph of the SC code (the graph of the SC code with all edge weights set to $1$), while the third stage aims at optimizing the edge weights. The three consecutive stages are:
\begin{enumerate}
\item We operate on the binary protograph of the SC code, and express the number of subgraphs we want to minimize in terms of the \textit{\textbf{overlap parameters}}, which characterize the partitioning of the block code. Then, we solve this discrete optimization problem to determine the optimal overlap parameters. We call this new partitioning technique the \textit{\textbf{optimal overlap (OO) partitioning}}. 
\item Given the optimal partitioning, we then apply a new heuristic program to optimize the \textit{\textbf{circulant powers}} of the underlying block code to further reduce the number of problematic subgraphs in the unlabeled graph of the SC code. We call this heuristic program the \textit{\textbf{circulant power optimizer (CPO)}}. 
\item Having optimized the underlying topology using the first two stages (OO-CPO), in the last stage, we focus on the edge weight processing in order to remove as many as possible of the remaining detrimental GASTs in the NB-SC code. To achieve this goal, we use the \textit{\textbf{WCM framework}} \cite{ahh_jsac, ahh_tit}. We also enumerate the minimum cardinality sets of edge weight changes that are candidates for the GAST removal.
\end{enumerate}
The three stages are necessary for the NB-SC code design procedure. We demonstrate the advantages of our code design approach over approaches that use CV partitioning and MO partitioning in the context of column weight $3$ SC codes.

The rest of the paper is organized as follows. In Section~\ref{sec_prelim}, we present some preliminaries. In Section \ref{sec_oo}, we detail the theory of the OO partitioning in the context of column weight $3$ SC codes. The CPO is then described in Section \ref{sec_cpo}. Next, in Section \ref{sec_wcm}, we propose a further discussion about the WCM framework. Our NB-SC code design steps and simulation results are presented in Section \ref{sec_sims}. Finally, the paper is concluded in Section \ref{sec_conc}.

\vspace{-0.3em}
\section{Preliminaries}\label{sec_prelim}
\vspace{-0.1em}

In this section, we review the construction of NB-SC~codes, as well as the CV and MO partitioning techniques. Furthermore, we recall the definition of GASTs and the key idea of the WCM framework.

Throughout this paper, each column (resp., row) in a parity-check matrix corresponds to a variable node (VN) (resp., check node (CN)) in the equivalent graph of the matrix.  Moreover, each non-zero entry in a parity-check matrix corresponds to an edge in the equivalent graph of the matrix.

Let $\bold{H}$ be the parity-check matrix of the underlying regular non-binary CB code that has column weight (VN degree) $\gamma$ and row weight (CN degree) $\kappa$. The binary image of $\bold{H}$, which is $\bold{H}^b$, consists of $\gamma \kappa$ circulants. Each circulant is of the form $\sigma^{f_{i, j}}$, where $i$, $0 \leq i \leq \gamma-1$, is the row group index, $j$, $0 \leq j \leq \kappa-1$, is the column group index, and $\sigma$ is the $p \times p$ identity matrix cyclically shifted one unit to the left (a circulant permutation matrix). Circulant powers are $f_{i, j}$, $\forall i$ and $\forall j$. Array-based (AB) codes are CB codes with $f_{i, j} = ij$, $\kappa = p$, and $p$ prime. In this paper, the underlying block codes we use to design SC codes are CB codes with no zero circulants.

The NB-SC code is constructed as follows. First, $\bold{H}^b$ is partitioned into $m+1$ disjoint components (of the same size as $\bold{H}^b$): $\bold{H}^b_0, \bold{H}^b_1, \dots, \bold{H}^b_m$, where $m$ is defined as the memory of the SC code. Each component $\bold{H}^b_y$, $0 \leq y \leq m$, contains some of the $\gamma \kappa$ circulants of $\bold{H}^b$ and zero circulants elsewhere such that $\bold{H}^b = \sum_{y=0}^{m} \bold{H}^b_y$. In this work, we focus on $m=1$, i.e., $\bold{H}^b = \bold{H}^b_0 + \bold{H}^b_1$. Second, $\bold{H}^b_0$ and $\bold{H}^b_1$ are coupled together $L$ times (see \cite{mitch_sc} and \cite{homa_sc}) to construct the binary image of the parity-check matrix of the NB-SC code, $\bold{H}^b_{SC}$, which is of size $(L+1)\gamma p \times L\kappa p$. A \textit{\textbf{replica}} is any $(L+1)\gamma p \times \kappa p$ submatrix of $\bold{H}^b_{SC}$ that contains $\left [\bold{H}^{bT}_0 \text{ } \bold{H}^{bT}_1 \right ]^T$ and zero circulants elsewhere (see \cite{homa_mo}). Replicas are denoted by $\bold{R}_r$, $1 \leq r \leq L$. Overlap parameters for partitioning as well as circulant powers can be selected to enhance the properties of $\bold{H}^b_{SC}$. Third, the matrix $\bold{H}$ is generated by replacing each $1$ in $\bold{H}^b$ with a value $\in$ GF($q$)$\backslash \{0\}$ (we focus on $q=2^\lambda \geq 4$). Fourth, the parity-check matrix of the NB-SC code, $\bold{H}_{SC}$, is constructed by applying the partitioning and coupling scheme described above to $\bold{H}$.

The \textit{\textbf{binary protograph matrix (BPM)}} of a general binary CB matrix is the matrix resulting from replacing each $p \times p$ non-zero circulant with $1$, and each $p \times p$ zero circulant with $0$. The BPMs of $\bold{H}^b$, $\bold{H}^b_0$, and $\bold{H}^b_1$ are $\bold{H}^{bp}$, $\bold{H}^{bp}_0$, and $\bold{H}^{bp}_1$, respectively, and they are all of size $\gamma \times \kappa$. The BPM of $\bold{H}^b_{SC}$ is $\bold{H}^{bp}_{SC}$, and it is of size $(L+1)\gamma \times L\kappa$. This $\bold{H}^{bp}_{SC}$ also has $L$ replicas, $\bold{R}_r$, $1 \leq r \leq L$, but with $1 \times 1$ circulants.

A technique for partitioning $\bold{H}^b$ to construct $\bold{H}^b_{SC}$ is the CV partitioning \cite{mitch_sc, homa_sc}. In this technique, a vector of ascending non-negative integers, $\boldsymbol{\zeta} = [\zeta_0 \text{ } \zeta_1 \text{ } \dots \text{ } \zeta_{\gamma-1}]$, is used to partition $\bold{H}^b$ into $\bold{H}^b_0$ and $\bold{H}^b_1$. The matrix $\bold{H}^b_0$ has all the circulants in $\bold{H}^b$ with the indices $\{(i, j): j<\zeta_i\}$, and zero circulants elsewhere, and the matrix $\bold{H}^b_1$ is $\bold{H}^b-\bold{H}^b_0$. Another recently introduced partitioning technique is the MO partitioning \cite{homa_mo}, in which $\bold{H}^b$ is partitioned into $\bold{H}^b_0$ and $\bold{H}^b_1$ such that the overlap of each pair of rows of circulants in both $\bold{H}^b_0$ and $\bold{H}^b_1$ is minimized. Moreover, the MO partitioning assumes balanced partitioning between $\bold{H}^b_0$ and $\bold{H}^b_1$, and also balanced distribution of circulants among the rows in each of them. The MO partitioning significantly outperforms the CV partitioning \cite{homa_mo}. In this paper, we demonstrate that the new OO-CPO technique outperforms the MO technique.

GASTs are the objects that dominate the error floor of NB codes on asymmetric channels, e.g., practical Flash channels. We recall the definitions of GASTs and unlabeled GASTs.

\begin{definition}\label{def_gast}
(cf. \cite{ahh_jsac}) Consider a subgraph induced by a subset $\mathcal{V}$ of VNs in the Tanner graph of an NB code. Set all the VNs in $\mathcal{V}$ to values $\in$ GF($q$)$\backslash \{0\}$ and set all other VNs to $0$. The set $\mathcal{V}$ is said to be an $(a, b, d_1, d_2, d_3)$ \textbf{general absorbing set of type two (GAST)} over GF($q$) if the size of $\mathcal{V}$ is $a$, the number of unsatisfied CNs connected to $\mathcal{V}$ is $b$, the number of degree-$1$ (resp., $2$ and $> 2$) CNs connected to $\mathcal{V}$ is $d_1$ (resp., $d_2$ and  $d_3$), $d_2 > d_3$, all the unsatisfied CNs connected to $\mathcal{V}$ (if any) have either degree $1$ or degree $2$, and each VN in $\mathcal{V}$ is connected to strictly more satisfied than unsatisfied neighboring CNs (for some set of given VN values).
\end{definition}

\begin{definition}\label{def_ugas}
(cf. \cite{ahh_jsac}) Let  $\mathcal{V}$ be a subset of VNs in the unlabeled Tanner graph of an NB code. Let $\mathcal{O}$ (resp., $\mathcal{T}$ and $\mathcal{H}$) be the set of degree-$1$ (resp., $2$ and $> 2$) CNs connected to $\mathcal{V}$. This graphical configuration is an $(a, d_1, d_2, d_3)$ \textbf{unlabeled GAST (UGAST)} if it satisfies the following two conditions:
\vspace{-0.2em}
\begin{enumerate}
\item $|\mathcal{V}| = a$, $\vert{\mathcal{O}}\vert=d_1$, $\vert{\mathcal{T}}\vert=d_2$, $\vert{\mathcal{H}}\vert=d_3$, and $d_2 > d_3$.
\item Each VN in $\mathcal{V}$ is connected to strictly more neighbors in $\{\mathcal{T} \cup \mathcal{H}\}$ than in $\mathcal{O}$.
\end{enumerate}
\end{definition}

Examples on GASTs and UGASTs are shown in Fig. \ref{Fig_gasts}.

The WCM framework \cite{ahh_jsac, ahh_tit} removes a GAST by careful processing of its edge weights. The key idea of this framework is to represent the GAST in terms of a set of submatrices of the GAST adjacency matrix. These submatrices are the WCMs, and they have the property that once the edge weights of the GAST are processed to force the null spaces of the WCMs to have a particular property, the GAST is completely removed from the Tanner graph of the NB code (see \cite{ahh_jsac} and \cite{ahh_tit}).

\vspace{-0.1em}
\section{OO Partitioning: Theoretical Analysis}\label{sec_oo}
\vspace{-0.1em}

In order to simultaneously reduce the number of multiple UGASTs, we determine a common substructure in them, then minimize the number of instances of this substructure in the unlabeled Tanner graph of the SC code (the graph of $\bold{H}^b_{SC}$) \cite{homa_sc}. We propose our new partitioning scheme in the context of SC codes with $\gamma = 3$ (the scheme can be extended to higher column weights). For the overwhelming majority of dominant GASTs we have encountered in NB codes with $\gamma = 3$ simulated over Flash channels, the $(3, 3, 3, 0)$ UGAST occurs as a common substructure most frequently \cite{ahh_jsac, ahh_tit} (see Fig.~\ref{Fig_gasts}). Thus, we focus on the removal of $(3, 3, 3, 0)$ UGASTs.

\begin{figure}[H]
\vspace{-1.5em}
\center
\includegraphics[width=3.3in]{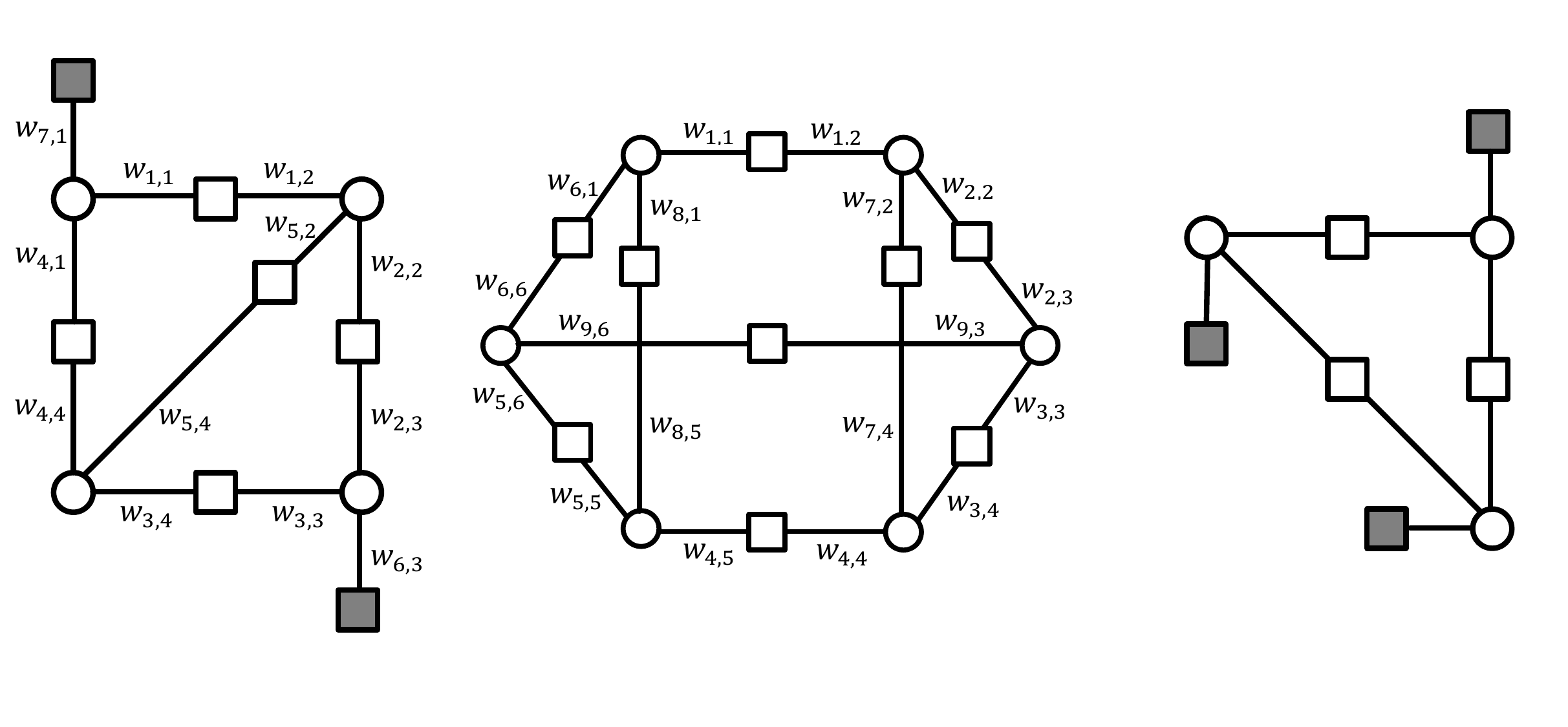}\vspace{-1.6em}
\vspace{-0.5em}
\text{\hspace{4em}\footnotesize{(a) \hspace{14em} (b)}}
\caption{(a) Two dominant GASTs for NB codes with $\gamma = 3$ over Flash; a $(4, 2, 2, 5, 0)$ and a $(6, 0, 0, 9, 0)$ GASTs. Appropriate edge weights ($w$'s) are assumed. (b) A $(3, 3, 3, 0)$ UGAST ($\gamma=3$).}
\label{Fig_gasts}
\vspace{-0.5em}
\end{figure}

A cycle of length $2z$ in the graph of $\bold{H}^{bp}_{SC}$ (the binary protograph of the SC code), which is defined by the non-zero entries $\{(h_1, \ell_1), (h_2, \ell_2), \dots, (h_{2z}, \ell_{2z})\}$ in $\bold{H}^{bp}_{SC}$, results in $p$ cycles of length $2z$ in the graph of $\bold{H}^b_{SC}$ if and only if \cite{baz_qc, fos_qc}:
\vspace{-0.1em}
\begin{align}\label{eq_cycle}
\sum_{e=1}^{z} f_{h_{2e-1}, \ell_{2e-1}} \equiv \sum_{e=1}^{z} f_{h_{2e}, \ell_{2e}} \text{ } (\text{mod } p),
\end{align}
where $f_{h, \ell}$ is the power of the circulant indexed by $(h, \ell)$ in $\bold{H}^b_{SC}$. Otherwise, this cycle results in $p/\beta$ cycle(s) of length $2z\beta$ in the graph of $\bold{H}^b_{SC}$, where $\beta$ is an integer $\geq 2$ that divides $p$ \cite{baz_qc}. It is clear from Fig. \ref{Fig_gasts}(b) that the $(3, 3, 3, 0)$ UGAST is a cycle of length $6$. Thus, and motivated by the above fact, our OO partitioning aims at deriving the overlap parameters of $\bold{H}^{bp}$ that result in the minimum number of cycles of length $6$ in the graph of $\bold{H}^{bp}_{SC}$, which is the binary protograph of the SC code. Then, we run the CPO to further reduce the number of $(3, 3, 3, 0)$ UGASTs in the graph of $\bold{H}^{b}_{SC}$ (which is the unlabeled graph of the SC code) by breaking the condition in (\ref{eq_cycle}) (with $z=3$ for cycles of length $6$) for as many cycles in the optimized graph of $\bold{H}^{bp}_{SC}$ as possible.

The goal here is to minimize the number of cycles of length $6$ in the binary protograph of the SC code via the OO partitioning of $\bold{H}^{bp}$, which is also the OO partitioning of $\bold{H}^b$. To achieve this goal, we establish a discrete optimization problem by expressing the number of cycles of length $6$ in the graph of $\bold{H}^{bp}_{SC}$ as a function of the overlap parameters and standard code parameters, then solve for the optimal overlap parameters. We start off with the following lemma.

\vspace{-0.2em}
\begin{lemma}\label{lem_total}
In the Tanner graph of an SC code with parameters $\gamma = 3$, $\kappa$, $p=1$, $m=1$, and $L$ (which is the binary protograph), the number of cycles of length $6$ is given by:
\begin{equation}\label{eq_Ftot}
F = LF_s + (L-1)F_d,\vspace{-0.1em}
\end{equation}
where $F_s$ is the number of cycles of length $6$ that have their VNs spanning only one particular replica (say $\bold{R}_1$), and $F_d$ is the number of cycles of length $6$ that have their VNs spanning two particular consecutive replicas (say $\bold{R}_1$ and $\bold{R}_2$).
\end{lemma}

\begin{proof}
From \cite[Lemma 1]{homa_mo}, the maximum number of consecutive replicas spanned by the $3$ VNs of a cycle of length $6$ in an SC code with $m = 1$ is $2$. Thus, the VNs of any cycle of length $6$ span either one replica or two consecutive replicas. Since there exist $L$ replicas and $L-1$ distinct pairs of consecutive replicas, and because of the repetitive nature of the SC code, (\ref{eq_Ftot}) follows.
\end{proof}
\vspace{-0.4em}

Let the \textit{\textbf{overlapping set}} of $x$ rows of a binary matrix be the set of positions in which all the $x$ rows have $1$'s simultaneously (overlap). Now, define the overlap parameters as follows:
\begin{itemize}
\item $t_i$ (resp., $t_{i+3}$), $0 \leq i \leq 2$, is the number of $1$'s in row $i$ of $\bold{H}^{bp}_0$ (resp., $\bold{H}^{bp}_1$). From the definitions of $\bold{H}^{bp}_0$ and $\bold{H}^{bp}_1$, $t_{i+3}=\kappa-t_i$.

\item $t_{i_1,i_2}$, $0 \leq i_1 \leq 2$, $0 \leq i_2 \leq 2$, and $i_2 > i_1$, is the size of the overlapping set of rows $i_1$ and $i_2$ of $\bold{H}^{bp}_0$.

\item $t_{i_3,i_4}$, $i_3=i_1+3$, $i_4=i_2+3$, and $i_4 > i_3$, is the size of the overlapping set of rows $i_1$ and $i_2$ of $\bold{H}^{bp}_1$. From the definitions, $t_{i_3,i_4}=\kappa-t_{i_1}-t_{i_2}+t_{i_1,i_2}$.

\item $t_{0,1,2}$ (resp., $t_{3,4,5}$) is the size of the overlapping set of rows $0$, $1$, and $2$ of $\bold{H}^{bp}_0$ (resp., $\bold{H}^{bp}_1$). Moreover, $t_{3,4,5} = \kappa - (t_0+t_1+t_2) + (t_{0,1}+t_{0,2}+t_{1,2})-t_{0,1,2}$.
\end{itemize}
Let $[x]^+=\max(x, 0)$. We define the following functions to be used in Theorem \ref{th_fsfd}:
\vspace{-0.1em}
\begin{align}
&\mathcal{A}(t_{0,1}, t_{0,2}, t_{1,2}, t_{0,1,2}) = \left [ t_{0,1,2}(t_{0,1,2}-1)(t_{1,2}-2) \right ]^+ \nonumber \\ &+ \left [ t_{0,1,2}(t_{0,2}-t_{0,1,2})(t_{1,2}-1) \right ]^+ \nonumber \\ &+ \left [ (t_{0,1}-t_{0,1,2})t_{0,1,2}(t_{1,2}-1) \right ] ^+ \nonumber \\ &+ \left [ (t_{0,1}-t_{0,1,2})(t_{0,2}-t_{0,1,2})t_{1,2} \right ]^+,
\end{align}
\begin{align}
&\mathcal{B}(t_0, t_1, t_2, t_{0,1}, t_{0,2}, t_{1,2}, t_{0,1,2}) \nonumber \\ &= \left [ t_{0,1,2}(t_{0,1}-t_{0,1,2})(t_1-t_{1,2}-1) \right ]^+ \nonumber \\ &+ \left [ t_{0,1,2}(t_0-t_{0,1}-t_{0,2}+t_{0,1,2})(t_1-t_{1,2}) \right ]^+ \nonumber \\ &+ \left [ (t_{0,1}-t_{0,1,2})(t_{0,1}-t_{0,1,2}-1)(t_1-t_{1,2}-2) \right ]^+ \nonumber \\ &+ \left [ (t_{0,1}-t_{0,1,2})(t_0-t_{0,1}-t_{0,2}+t_{0,1,2})(t_1-t_{1,2}-1) \right ]^+ \nonumber \\ &+ \left [ t_{0,1,2}(t_{0,2}-t_{0,1,2})(t_0-t_{0,1}-1) \right ]^+ \nonumber \\ &+ \left [ t_{0,1,2}(t_2-t_{0,2}-t_{1,2}+t_{0,1,2})(t_0-t_{0,1}) \right ]^+ \nonumber \\ &+ \left [ (t_{0,2}-t_{0,1,2})(t_{0,2}-t_{0,1,2}-1)(t_0-t_{0,1}-2) \right ]^+ \nonumber \\ &+ \left [ (t_{0,2}-t_{0,1,2})(t_2-t_{0,2}-t_{1,2}+t_{0,1,2})(t_0-t_{0,1}-1) \right ]^+ \nonumber \\ &+ \left [ t_{0,1,2}(t_{1,2}-t_{0,1,2})(t_2-t_{0,2}-1) \right ]^+ \nonumber \\ &+ \left [ t_{0,1,2}(t_1-t_{0,1}-t_{1,2}+t_{0,1,2})(t_2-t_{0,2}) \right ]^+ \nonumber \\ &+ \left [ (t_{1,2}-t_{0,1,2})(t_{1,2}-t_{0,1,2}-1)(t_2-t_{0,2}-2) \right ]^+ \nonumber \\ &+ \left [ (t_{1,2}-t_{0,1,2})(t_1-t_{0,1}-t_{1,2}+t_{0,1,2})(t_2-t_{0,2}-1) \right ]^+, \\
&\mathcal{C}(\kappa, t_0, t_1, t_2, t_{0,1}, t_{0,2}, t_{1,2}, t_{0,1,2}) \nonumber \\ &= \left [ t_{3,4}t_{0,1,2}(t_{1,2}-1) \right ]^+ + \left [ t_{3,4}(t_{0,2}-t_{0,1,2})t_{1,2} \right ]^+ \nonumber \\ &+ \left [ t_{3,5}t_{0,1,2}(t_{0,1}-1) \right ]^+ + \left [ t_{3,5}(t_{1,2}-t_{0,1,2})t_{0,1} \right ]^+ \nonumber \\ &+ \left [ t_{4,5}t_{0,1,2}(t_{0,2}-1) \right ]^+ + \left [ t_{4,5}(t_{0,1}-t_{0,1,2})t_{0,2} \right ]^+, \textit{ and} \\
&\mathcal{D}(t_0, t_1, t_2, t_{0,1}, t_{0,2}, t_{1,2}, t_{0,1,2}) \nonumber \\ &= \left [ t_{0,1}(t_2-t_{0,2}-t_{1,2}+t_{0,1,2})(t_2-t_{1,2}-1) \right ]^+ \nonumber \\ &+ \left [ t_{0,1}(t_{1,2}-t_{0,1,2})(t_2-t_{1,2}) \right ]^+ \nonumber \\ & + \left [ t_{0,2}(t_1-t_{0,1}-t_{1,2}+t_{0,1,2})(t_1-t_{0,1}-1) \right ]^+ \nonumber \\ &+ \left [ t_{0,2}(t_{0,1}-t_{0,1,2})(t_1-t_{0,1}) \right ]^+ \nonumber \\ &+ \left [ t_{1,2}(t_0-t_{0,1}-t_{0,2}+t_{0,1,2})(t_0-t_{0,2}-1) \right ]^+ \nonumber \\ &+ \left [ t_{1,2}(t_{0,2}-t_{0,1,2})(t_0-t_{0,2}) \right ]^+.
\end{align}
\vspace{-1.3em}

Theorem \ref{th_fsfd} uses combinatorics to give the exact expressions for $F_s$ and $F_d$ in terms of the above overlap parameters.

\vspace{-0.1em}
\begin{thm}\label{th_fsfd}
In the Tanner graph of an SC code with parameters $\gamma = 3$, $\kappa$, $p=1$, $m=1$, and $L$ (which is the binary protograph), $F_s$ and $F_d$ are computed as follows:
\vspace{-0.2em}\begin{align}
F_s &= F_{s,0} + F_{s,1} + F_{s,2} + F_{s,3}, \textit{ and} \\ \vspace{-0.2em}
F_d &= F_{d,0} + F_{d,1} + F_{d,2} + F_{d,3},\vspace{-0.6em}
\end{align}
where $F_{s,0}$, $F_{s,1}$, $F_{s,2}$, $F_{s,3}$, $F_{d,0}$, $F_{d,1}$, $F_{d,2}$, and $F_{d,3}$ are:
\vspace{-0.3em}\begin{align}
F_{s,0} &= \mathcal{A}(t_{0,1}, t_{0,2}, t_{1,2}, t_{0,1,2}), \nonumber \\ 
F_{s,1} &= \mathcal{A}(t_{3,4}, t_{3,5}, t_{4,5}, t_{3,4,5}), \nonumber \\
F_{s,2} &= \mathcal{B}(t_0, t_1, t_2, t_{0,1}, t_{0,2}, t_{1,2}, t_{0,1,2}), \nonumber \\
F_{s,3} &= \mathcal{B}(t_3, t_4, t_5, t_{3,4}, t_{3,5}, t_{4,5}, t_{3,4,5}), \\
F_{d,0} &= \mathcal{C}(\kappa, t_0, t_1, t_2, t_{0,1}, t_{0,2}, t_{1,2}, t_{0,1,2}), \nonumber \\ 
F_{d,1} &= \mathcal{C}(\kappa, t_3, t_4, t_5, t_{3,4}, t_{3,5}, t_{4,5}, t_{3,4,5}), \nonumber \\
F_{d,2} &= \mathcal{D}(t_0, t_1, t_2, t_{0,1}, t_{0,2}, t_{1,2}, t_{0,1,2}), \textit{ and} \nonumber \\
F_{d,3} &= \mathcal{D}(t_3, t_4, t_5, t_{3,4}, t_{3,5}, t_{4,5}, t_{3,4,5}).
\end{align}
\end{thm}

\begin{proof}
The term $F_s$ represents the number of cycles of length $6$ that have their VNs spanning only one replica. The non-zero submatrix of a replica is $\left [\bold{H}_0^{bpT} \text{ } \bold{H}_1^{bpT} \right ]^T$. There are four possible cases of arrangement for the CNs of a cycle of length $6$ that has its VNs spanning only one replica. These cases are listed below:
\begin{enumerate}
\item All the three CNs are within $\bold{H}_0^{bp}$. The number of cycles of length $6$ that have all their CNs inside $\bold{H}_0^{bp}$ is denoted by $F_{s,0}$.
\item All the three CNs are within $\bold{H}_1^{bp}$. The number of cycles of length $6$ that have all their CNs inside $\bold{H}_1^{bp}$ is denoted by $F_{s,1}$.
\item Two CNs are within $\bold{H}_0^{bp}$, and one CN is within $\bold{H}_1^{bp}$. The number of cycles of length $6$ in this case is denoted by $F_{s,2}$.
\item Two CNs are within $\bold{H}_1^{bp}$, and one CN is within $\bold{H}_0^{bp}$. The number of cycles of length $6$ in this case is denoted by $F_{s,3}$.
\end{enumerate}
These four different cases of arrangement are illustrated in the upper panel of Fig. \ref{Fig_oo}. Next, we find the number of cycles of length $6$ in each of the four cases in terms of the overlap parameters and standard code parameters, particularly, $\{\kappa,t_0,t_1,t_2,t_{0,1},t_{0,2},t_{1,2},t_{0,1,2}\}$.

In case 1, a cycle of length $6$ is comprised of an overlap between rows $0$ and $1$, an overlap between rows $0$ and $2$, and an overlap between rows $1$ and $2$ of $\bold{H}_0^{bp}$. Note that each overlap must have a distinct associated column index (position) to result in a valid cycle of length $6$. The overlap between rows $0$ and $1$ can be selected among $t_{0,1}$ possible choices. Among these $t_{0,1}$ overlaps, there exist $t_{0,1,2}$ overlaps that have the same associated column indices as some overlaps between other pairs of rows. Thus, these $t_{0,1,2}$ overlaps need to be considered separately to avoid incorrect counting. The same argument applies when we choose the overlap between the other two pairs of rows. As a result, the number of different ways to choose these overlaps and form a cycle of length $6$ is $F_{s,0}=\mathcal{A}(t_{0,1},t_{0,2},t_{1,2},t_{0,1,2})$, and $\mathcal{A}$ is defined in (3).

In case 2, the number of cycles of length $6$, $F_{s,1}$, is computed exactly as in case 1, but using the overlap parameters of the matrix $\bold{H}_1^{bp}$. Thus, $F_{s,1}=\mathcal{A}(t_{3,4},t_{3,5},t_{4,5},t_{3,4,5})$.

In case 3, one overlap solely belongs to $\bold{H}_0^{bp}$, and the two other overlaps cross $\bold{H}_0^{bp}$ to $\bold{H}_1^{bp}$ (see Fig. \ref{Fig_oo}). For the overlap in $\bold{H}_0^{bp}$, we have three options to choose two rows out of three. For example, suppose that the overlap is chosen between rows $0$ and $1$ of $\bold{H}_0^{bp}$. Then, the cross overlaps will be between row $0$ of $\bold{H}_0^{bp}$ and row $2$ of $\bold{H}_1^{bp}$, and also between row $1$ of $\bold{H}_0^{bp}$ and row $2$ of $\bold{H}_1^{bp}$. Note that since $\bold{H}_0^{bp}$ and $\bold{H}_1^{bp}$ are the result of partitioning $\bold{H}^{bp}$, there are no overlaps between row $i$ of $\bold{H}_0^{bp}$ and row $i$ of $\bold{H}_1^{bp}$, $0 \leq i \leq 2$. Based on which option of the three is chosen, the number of cycles of length $6$ is computed using the overlap parameters of $\bold{H}_0^{bp}$ and $\bold{H}_1^{bp}$. The total number of cycles of length $6$ in this case is $F_{s,2}=\mathcal{B}(t_0,t_1,t_2,t_{0,1},t_{0,2},t_{1,2},t_{0,1,2})$, and $\mathcal{B}$ is defined in (4).

In case 4, the number of cycles of length $6$, $F_{s,3}$, is computed as in case 3. The only difference is that in case 4, one overlap solely belongs to $\bold{H}_1^{bp}$, and the two other overlaps cross $\bold{H}_0^{bp}$ to $\bold{H}_1^{bp}$ (see Fig. \ref{Fig_oo}). Consequently, $F_{s,3}=\mathcal{B}(t_3,t_4,t_5,t_{3,4},t_{3,5},t_{4,5},t_{3,4,5})$.

On the other hand, the term $F_d$ represents the number of cycles of length $6$ that have their VNs spanning two consecutive replicas. The non-zero submatrix of two consecutive replicas is:
\begin{equation*}
\left[\begin{array}{cc}
\bold{H}_0^{bp}&\bold{0}\\
\bold{H}_1^{bp}&\bold{H}_0^{bp}\\
\bold{0}&\bold{H}_1^{bp}
\end{array}\right].
\end{equation*}
There are four possible cases of arrangement for the CNs and VNs of a cycle of length $6$ that has its VNs spanning two consecutive replicas. These cases are listed below:

\begin{enumerate}
\item All the three CNs are within $[\bold{H}_1^{bp} \text{ } \bold{H}_0^{bp}]$, two VNs belong to the first replica, and one VN belongs to the second replica. The number of cycles of length $6$ in this case is denoted by $F_{d,0}$.
\item All the three CNs are within $[\bold{H}_1^{bp} \text{ } \bold{H}_0^{bp}]$, one VN belongs to the first replica, and two VNs belong to the second replica. The number of cycles of length $6$ in this case is denoted by $F_{d,1}$.
\item One CN is within $[\bold{H}_0^{bp}\text{ }\bold{0}]$, and two CNs are within $[\bold{H}_1^{bp} \text{ } \bold{H}_0^{bp}]$. Besides, two VNs belong to the first replica, and one VN belongs to the second replica. The number of cycles of length $6$ in this case is denoted by $F_{d,2}$.
\item Two CNs are within $[\bold{H}_1^{bp} \text{ } \bold{H}_0^{bp}]$, and one CN is within $[\bold{0}\text{ }\bold{H}_1^{bp}]$. Besides, one VN belongs to the first replica, and two VNs belong to the second replica. The number of cycles of length $6$ in this case is denoted by $F_{d,3}$.
\end{enumerate}
These four different cases of arrangement are illustrated in the lower panel of Fig. \ref{Fig_oo}. Next, we find the number of cycles of length $6$ in each of the four cases in terms of the overlap parameters and standard code parameters, particularly, $\{\kappa,t_0,t_1,t_2,t_{0,1},t_{0,2},t_{1,2},t_{0,1,2}\}$.

In case 1, two overlaps belong to $\bold{H}_1^{bp}$ in the first replica, and one overlap belongs to $\bold{H}_0^{bp}$ in the second replica (see Fig.~\ref{Fig_oo}). For the overlap in $\bold{H}_0^{bp}$, we have three options to choose two rows out of three. For each option, the two overlaps inside $\bold{H}_1^{bp}$ must have distinct associated column indices (positions) to result in a valid cycle of length $6$ (the overlap inside $\bold{H}_0^{bp}$ cannot have the same column index as any of the other two overlaps). Thus, the number of different ways to choose these overlaps and form a cycle of length $6$ is given by $F_{d,0}=\mathcal{C}(\kappa,t_0,t_1,t_2,t_{0,1},t_{0,2},t_{1,2},t_{0,1,2})$, and $\mathcal{C}$ is defined in (5).

In case 2, the number of cycles of length $6$, $F_{d,1}$, is computed as in case 1. The only difference is that in case~2, one overlap belongs to $\bold{H}_1^{bp}$ in the first replica, and two overlaps belong to $\bold{H}_0^{bp}$ in the second replica (see Fig. \ref{Fig_oo}). Thus, $F_{d,1}=\mathcal{C}(\kappa,t_3,t_4,t_5,t_{3,4},t_{3,5},t_{4,5},t_{3,4,5})$.

In case 3, one overlap solely belongs to $\bold{H}_0^{bp}$ in the second replica, and the two other overlaps cross $\bold{H}_0^{bp}$ to $\bold{H}_1^{bp}$ in the first replica (see Fig. \ref{Fig_oo}). For the overlap in $\bold{H}_0^{bp}$ of the second replica, we have three options to choose two rows out of three. The two overlaps that belong to the first replica must have distinct corresponding column indices (positions). Consequently, the total number of cycles of length $6$ in this case is given by $F_{d,2}=\mathcal{D}(t_0,t_1,t_2,t_{0,1},t_{0,2},t_{1,2},t_{0,1,2})$, and $\mathcal{D}$ is defined in (6).

In case 4, the number of cycles of length $6$, $F_{d,3}$, is computed as in case 3. The only difference is that in case~4, one overlap solely belongs to $\bold{H}_1^{bp}$ in the first replica, and the two other overlaps cross $\bold{H}_0^{bp}$ to $\bold{H}_1^{bp}$ in the second replica (see Fig. \ref{Fig_oo}). Thus, $F_{d,3}=\mathcal{D}(t_3,t_4,t_5,t_{3,4},t_{3,5},t_{4,5},t_{3,4,5})$.

Note that the operator $[.]^+$ is used to avoid counting options that are not valid.
\end{proof}
\vspace{-0.2em}

\begin{figure}[H]
\vspace{-1.2em}
\center
\includegraphics[trim={1.3in 7.6in 1.2in 0.9in},clip,width=3.5in]{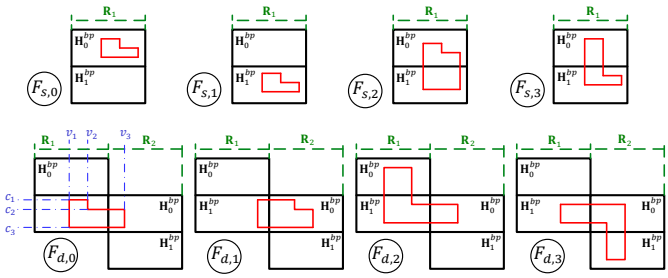}
\vspace{-1.5em}
\caption{Different cases for the cycle of length $6$ (in red) in a $\gamma=3$ SC binary protograph. The upper panel (resp., lower panel) is for the case of the VNs spanning $\bold{R}_1$ (resp., $\bold{R}_1$ and $\bold{R}_2$).}
\label{Fig_oo}
\vspace{-0.3em}
\end{figure}

The main idea of Theorem \ref{th_fsfd} is that both $F_s$ and $F_d$ can be computed by decomposing each of them into four more tractable terms. Each term represents a distinct case for the existence of a cycle of length $6$ in the SC binary protograph, and the union of these cases covers all the existence possibilities. Each case is characterized by the locations of the CNs and VNs comprising the cycle with respect to $\bold{H}^{bp}_0$ and $\bold{H}^{bp}_1$ of the replica $\bold{R}_1$ (for $F_s$) or the replicas $\bold{R}_1$ and $\bold{R}_2$ (for $F_d$). Fig. \ref{Fig_oo} illustrates these eight cases, along with the terms in $F_s$ and $F_d$ that corresponds to each case.

\vspace{-0.1em}
\begin{remark}
Consider the special situation of $t_{0,1,2}=0$ (rows $0$, $1$, and $2$ in $\bold{H}^{bp}_0$ do not have a $3$-way overlap). Here, $F_{s,0}$ reduces to $t_{0,1}t_{0,2}t_{1,2}$, which is simply the number of ways to select one position from the overlapping set of each pair.
\end{remark}

Now, define $F^*$ to be the minimum number of cycles of length $6$ in the graph of $\bold{H}^{bp}_{SC}$ (the binary protograph). Thus, our \textit{\textbf{discrete optimization problem}} is formulated as follows:
\begin{equation}\label{eq_optp}
F^* = \min_{t_0, t_1, t_2, t_{0,1}, t_{0,2}, t_{1,2}, t_{0,1,2}} F.
\end{equation}

The constraints of our optimization problem are the conditions under which the overlap parameters are valid. Thus, these constraints on the seven parameters in (\ref{eq_optp}) are:
\begin{align}\label{eq_const}
&0 \leq t_0 \leq \kappa, \text{ } 0 \leq t_{0,1} \leq t_0, \text{ } t_{0,1} \leq t_1 \leq \kappa-t_0+t_{0,1}, \nonumber \\
&0 \leq t_{0,1,2} \leq t_{0,1}, \text{ } t_{0,1,2} \leq t_{0,2} \leq t_0-t_{0,1}+t_{0,1,2}, \nonumber \\
&t_{0,1,2} \leq t_{1,2} \leq t_1-t_{0,1}+t_{0,1,2}, \nonumber \\
&t_{0,2}+t_{1,2}-t_{0,1,2} \leq t_2 \leq \kappa-t_0-t_1+t_{0,1}+t_{0,2}+t_{1,2} \nonumber \\ &-t_{0,1,2}, \textit{ and }\left \lfloor {3\kappa}/{2} \right \rfloor \leq t_0+t_1+t_2 \leq \left \lceil {3\kappa}/{2} \right \rceil.
\end{align}
The last constraint in (\ref{eq_const}) guarantees balanced partitioning between $\bold{H}^{bp}_{0}$ and $\bold{H}^{bp}_{1}$, and it is needed to prevent the case that a group of non-zero elements (a group of $1$'s) in either $\bold{H}^{bp}_{0}$ or $\bold{H}^{bp}_{1}$ are involved in significantly more cycles than the remaining non-zero elements (the remaining $1$'s). The solution of our optimization problem is not unique. However, since all the solutions result in the same number of OO partitioning choices and the same $F^*$, we work with one of these solutions, and call it an optimal vector, $\bold{t}^{*} = [t^*_0 \text{ } t^*_1 \text{ } t^*_2 \text{ } t^*_{0,1} \text{ } t^*_{0,2} \text{ } t^*_{1,2} \text{ } t^*_{0,1,2}]$.

Lemma \ref{lem_choi} gives the total number of OO partitioning choices.

\begin{lemma}\label{lem_choi}
The total number of OO partitioning choices for an SC code with parameters $\gamma = 3$, $\kappa$, $p=1$, $m=1$, and $L$ (which is the binary protograph) given an optimal vector $\bold{t}^*$ is given by:
\vspace{-0.4em}
\begin{align}\label{eq_choi}
\mathcal{N} = {\alpha} &\binom{\kappa}{t^*_0} \binom{t^*_{0}}{t^*_{0,1}} \binom{\kappa-t^*_{0}}{t^*_{1}-t^*_{0,1}} \binom{t^*_{0,1}}{ t^*_{0,1,2}} \binom{t^*_{0}-t^*_{0,1}}{t^*_{0,2}-t^*_{0,1,2}} \nonumber \\ &\binom{t^*_{1}-t^*_{0,1}}{t^*_{1,2}-t^*_{0,1,2}} \binom{\kappa-t^*_{0}-t^*_{1}+t^*_{0,1}}{t^*_{2}-t^*_{0,2}-t^*_{1,2}+t^*_{0,1,2}},
\end{align}
where $\alpha$ is the number of distinct solutions (optimal vectors).
\end{lemma}

\begin{proof}
The goal is to find the number of partitioning choices that achieve a general set of overlap parameters $\{t_0,t_1,t_2,t_{0,1}, \allowbreak t_{0,2},t_{1,2},t_{0,1,2}\}$ (not necessarily optimal). In particular, we need to find the number of different partitioning choices of an SC code with $\gamma=3$, $\kappa$, $p=1$, $m=1$, and $L$ such that:
\begin{itemize}
\item The number of $1$'s in row $i$, $0 \leq i \leq 2$, of $\bold{H}_0^{bp}$ is $t_i$.
\item The size of the overlapping set of rows $i_1$ and $i_2$, $0 \leq i_1 \leq 2$, $0 \leq i_2 \leq 2$, and $i_2 > i_1$, of $\bold{H}^{bp}_0$ is $t_{i_1,i_2}$.
\item The size of the overlapping set of rows $0$, $1$, and $2$ ($3$-way overlap) of $\bold{H}^{bp}_0$ is $t_{0,1,2}$.
\end{itemize}

We factorize the number of partitioning choices, $\mathcal{N}^g$, into three more tractable factors:
\begin{itemize}
\item[--] Choose $t_0$ positions, in which row $0$ of $\bold{H}_0^{bp}$ has $1$'s, out of $\kappa$ positions. The number of choices is:
$$\mathcal{N}^g_0={\kappa\choose{t_0}}.$$
\item[--] Choose $t_1$ positions, in which row $1$ of $\bold{H}_0^{bp}$ has $1$'s, out of $\kappa$ positions. Among these $t_1$ positions, there exist $t_{0,1}$ positions in which row $0$ simultaneously has $1$'s. The number of choices is:
$$\mathcal{N}^g_1={{t_0}\choose{t_{0,1}}}{{\kappa-t_0}\choose{t_1-t_{0,1}}}.$$
\item[--] Choose $t_2$ positions, in which row $2$ of $\bold{H}_0^{bp}$ has $1$'s, out of $\kappa$ positions. Among these $t_2$ positions, there exist $t_{0,1,2}$ positions in which rows $0$ and $1$ simultaneously have $1$'s, $t_{0,2}$ positions in which only rows $0$ simultaneously has $1$'s, and $t_{1,2}$ positions in which only rows $1$ simultaneously has $1$'s. The number of choices is:
\begin{equation*}
\begin{aligned}
\mathcal{N}^g_2=&{{t_{0,1}}\choose{t_{0,1,2}}}{{t_0-t_{0,1}}\choose{t_{0,2}-t_{0,1,2}}}{{t_1-t_{0,1}}\choose{t_{1,2}-t_{0,1,2}}}\\
&{{\kappa-t_0-t_1+t_{0,1}}\choose{t_2-t_{0,2}-t_{1,2}+t_{0,1,2}}}.
\end{aligned}
\end{equation*}
\end{itemize}
In conclusion, the number of partitioning choices that achieve a general set of overlap parameters is $\mathcal{N}^g_0\mathcal{N}^g_1\mathcal{N}^g_2$.

The solution of the optimization problem in (\ref{eq_optp}) is not unique, and there are $\alpha$ distinct solutions (optimal vectors) that all achieve $F^*$. Because of the symmetry of these $\alpha$ optimal vectors, each of them corresponds to the same $\mathcal{N}_0\mathcal{N}_1\mathcal{N}_2$ partitioning choices. The factors $\mathcal{N}_0$, $\mathcal{N}_1$, and $\mathcal{N}_2$ are obtianed by replacing each $t$ with $t^*$ (from an optimal vector $\bold{t}^*$) in the equations of $\mathcal{N}^g_0$, $\mathcal{N}^g_1$, and $\mathcal{N}^g_2$, respectively. Thus, the total number of OO partitioning choices given an optimal vector $\textbf{t}^*$ is $\mathcal{N}=\alpha\mathcal{N}_0\mathcal{N}_1\mathcal{N}_2$, which proves Lemma \ref{lem_choi}.
\end{proof}

\vspace{-0.8em}
\begin{remark}
The first seven constraints of the optimization problem, which are stated in (\ref{eq_const}), can be easily verified from (\ref{eq_choi}) in Lemma \ref{lem_choi} by replacing each $t^*$ with $t$.
\end{remark}

\vspace{-0.9em}
\section{Circulant Power Optimization}\label{sec_cpo}

After picking an optimal vector $\bold{t}^*$ to partition $\bold{H}^{bp}$ and design $\bold{H}^{bp}_{SC}$, we run our heuristic CPO to further reduce the number of $(3, 3, 3, 0)$ UGASTs in the graph of $\bold{H}^{b}_{SC}$, which has $\gamma=3$. The steps of the CPO are:
\begin{enumerate}
\item Initially, assign circulant powers as in AB codes to all the $\gamma \kappa$ $1$'s in $\bold{H}^{bp}$ (results in no cycles of length $4$ in $\bold{H}^{b}$ and $\bold{H}^b_{SC}$).
\item Design $\bold{H}^{bp}_{SC2}$ using $\bold{H}^{bp}$ and $\bold{t}^*$ such that $\bold{H}^{bp}_{SC2}$ contains only two replicas, $\bold{R}_1$ and $\bold{R}_2$. Circulant powers of the $1$'s in $\bold{H}^{bp}_{SC2}$ are copied from the $1$'s in $\bold{H}^{bp}$.
\item Locate all the cycles of lengths $4$ and $6$ in $\bold{H}^{bp}_{SC2}$.
\item Specify the cycles of length $6$ in $\bold{H}^{bp}_{SC2}$ that have (\ref{eq_cycle}) satisfied, and call them \textit{\textbf{active cycles}}. Let $2F^a_s$ (resp., $F^a_d$) be the number of active cycles having their VNs spanning only $\bold{R}_1$ or only $\bold{R}_2$ (resp., both $\bold{R}_1$ and $\bold{R}_2$).
\item Compute the number of $(3, 3, 3, 0)$ UGASTs in $\bold{H}^{b}_{SC}$ using the following formula:
\vspace{-0.2em}\begin{equation}
F_{SC} = \left ( LF^a_s + (L-1)F^a_d \right ) p. \vspace{-0.2em}
\end{equation}
\item Count the number of active cycles each $1$ in $\bold{H}^{bp}_{SC2}$ is involved in. Give weight $1$ (resp., $2$) to the number of active cycles having their VNs spanning only $\bold{R}_1$ or only $\bold{R}_2$ (resp., both $\bold{R}_1$ and $\bold{R}_2$).
\item Map the counts from step 6 to the $1$'s in $\bold{H}^{bp}$, and sort these $1$'s in a list descendingly according to the counts.
\item Pick a subset of $1$'s from the top of this list, and change the circulant powers associated with them.
\item Using these interim new powers, do steps 4 and 5.
\item If $F_{SC}$ is reduced while maintaining no cycles of length $4$ in $\bold{H}^{b}_{SC}$, update $F_{SC}$ and the circulant powers, then go to step 6. Otherwise, return to step 8.
\item Iterate until the target $F_{SC}$ is achieved.
\end{enumerate}
Note that step 8 is performed heuristically.

\begin{figure}[H]
\vspace{-1.8em}
\center
\includegraphics[width=3.6in]{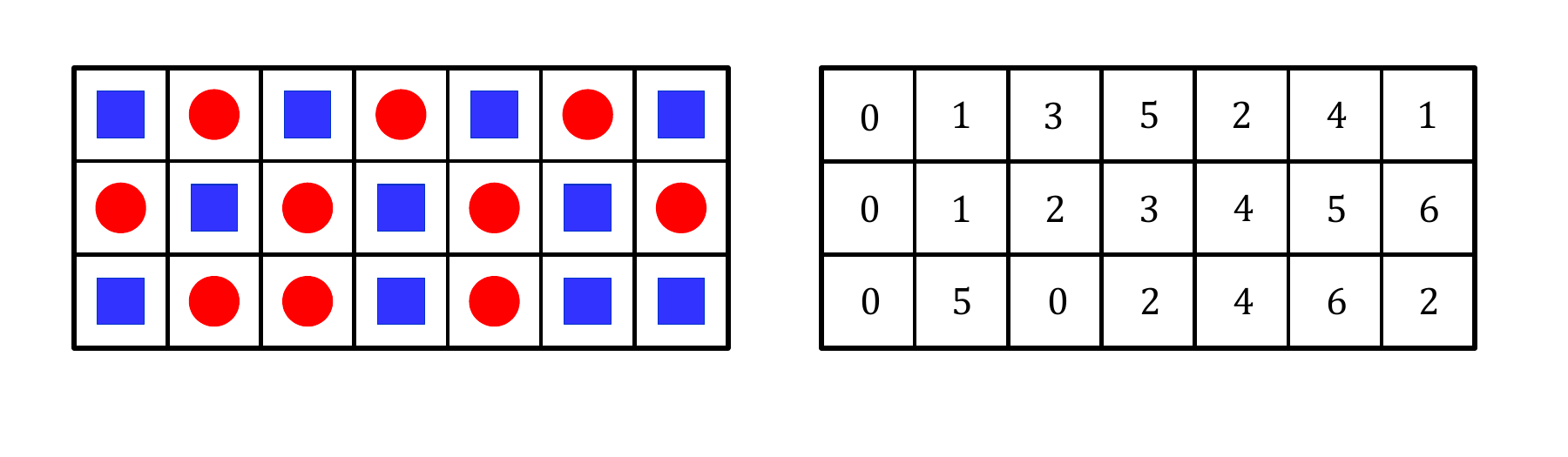}\vspace{-1.7em}
\vspace{-0.4em}
\text{\hspace{-0em}\footnotesize{(a) \hspace{14em} (b)}}
\caption{(a) The OO partitioning of $\bold{H}^{bp}$ (or $\bold{H}^{b}$) of the SC code in Example~\ref{ex_oo}. Entries with circles (resp., squares) are assigned to $\bold{H}^{bp}_0$ (resp., $\bold{H}^{bp}_1$). (b) The circulant power arrangement for the circulants in $\bold{H}^b$.}
\label{Fig_OOp7}
\vspace{-0.7em}
\end{figure}

\begin{example}\label{ex_oo}
Suppose we want to design an SC code with $\gamma = 3$, $\kappa=7$, $p=7$, $m=1$, and $L=30$ using the OO partitioning and the CPO. Solving the optimization problem in (\ref{eq_optp}) yields an optimal vector $\bold{t}^*=[3 \text{ } 4 \text{ } 3 \text{ } 0 \text{ } 1 \text{ } 2 \text{ } 0]$, which gives $F^* = 1170$ cycles of length $6$ in the graph of $\bold{H}^{bp}_{SC}$. Fig. \ref{Fig_OOp7}(a) shows how the partitioning is applied on $\bold{H}^{bp}$ (or $\bold{H}^{b}$). Next, applying the CPO results in only $203$ $(3, 3, 3, 0)$ UGASTs in the unlabeled graph of the SC code, which is the graph of $\bold{H}^{b}_{SC}$. Fig. \ref{Fig_OOp7}(b) shows the final circulant power arrangement for all circulants in $\bold{H}^b$.
\end{example}

The OO-CPO technique for designing $\bold{H}^b_{SC}$ is based on solving a set of equations, then applying a heuristic program on two replicas to optimize the circulant powers. Moreover, the OO partitioning has orders of magnitude fewer number of partitioning choices compared to the MO partitioning (see \cite[Lemma 3]{homa_mo}). We can even use any choice of the OO partitioning choices without having to compare their performances explicitly. All these reasons demonstrate that the OO-CPO technique is not only better in performance (see Section \ref{sec_sims} for details), but also much faster than the MO technique.

\vspace{-0.2em}
\section{WCM Framework: On The Removal of GASTs}\label{sec_wcm}

After applying the OO-CPO technique to optimize the unlabeled graph of the SC code, we optimize the edge weights. In particular, we use the WCM framework \cite{ahh_jsac, ahh_tit} to remove GASTs from the labeled graph of the NB-SC code through edge weight processing. There are multiple parameters that control the difficulty of the removal of a certain GAST from the Tanner graph of a code. The number of distinct WCMs associated with the UGAST and the minimum number of edge weight changes needed to remove the GAST, denoted by $E_{GAST,min}$, are among these parameters. A third parameter is the number of sets of edge weight changes that have cardinality $E_{GAST,min}$ and are candidates for the GAST removal process. The first two parameters are studied in \cite{ahh_tit}. We discuss the third parameter in this section. As the number of candidate sets of cardinality $E_{GAST,min}$ increases, the difficulty of the GAST removal decreases.

In this section, unless otherwise stated, when we say nodes are ``connected'', we mean they are ``directly connected'' or they are ``neighbors''. The same applies conceptually when we say an edge is ``connected'' to a node or vice versa.

\begin{remark}
A GAST is removed by performing $E_{GAST,min}$ edge weight changes for edges connected to degree-$2$ CNs only (see also \cite{ahh_tit}). Whether a candidate set of edge weight changes indeed results in the GAST removal or not is determined by checking the null spaces of the WCMs \cite{ahh_jsac, ahh_tit}.
\end{remark}

To minimize the number of edge weight changes performed to remove a GAST, we need to work on the VNs that are connected to the maximum number of unsatisfied CNs. Thus, $E_{GAST,min}=g-b_{vm}+1$ (see \cite{ahh_tit}), where $g = \left \lfloor \frac{\gamma-1}{2} \right  \rfloor$ and $b_{vm}$ is the maximum number of existing unsatisfied CNs per VN in the GAST. Define $E_{mu}$ as the topological upper bound on $E_{GAST,min}$ and $d_{1,vm}$ as the maximum number of existing degree-$1$ CNs per VN in the GAST. Thus, from \cite{ahh_tit}:
\begin{equation}\label{eq_emu}
E_{mu} = g-d_{1,vm}+1 \geq E_{GAST,min}.
\end{equation}
Note that (\ref{eq_emu}) follows from $d_{1,vm} \leq b_{vm}$. In this section, we study GASTs with $b=d_1$, which means the upper bound is achieved, i.e., $E_{GAST,min}=E_{mu}$. Moreover, for simplicity, we assume that all the VNs that are connected to $d_{1,vm}$ degree-$1$ CNs each are only connected to CNs of degree $\leq 2$.

\begin{thm}\label{th_choices}
Consider an $(a, b, d_1, d_2, d_3)$ GAST, with $b=d_1$, in an NB code defined over $GF($q$)$ that has column weight $\gamma$ and no cycles of length $4$. The number of sets of edge weight changes with cardinality $E_{GAST,min}$ (or $E_{mu}$) that are candidates for the GAST removal process is given as follows.

If $d_{1,vm} \neq g$:
\vspace{-0.3em}\begin{equation}\label{eq_choices1}
S_{mu} = a_{vm}\binom{\gamma-d_{1,vm}}{E_{mu}}(2(q-2))^{E_{mu}},
\end{equation}
where $a_{vm}$ is the number of VNs connected to $d_{1,vm}$ degree-$1$ CNs each.

If $d_{1,vm} = g$:
\vspace{-0.4em}\begin{equation}\label{eq_choices2}
S_{mu} = \left ( a_{vm} \left \lceil \frac{\gamma+1}{2} \right \rceil - n_{co} \right )2(q-2),
\end{equation}
where $n_{co}$ is the number of degree-$2$ CNs connecting any two of these $a_{vm}$ VNs.
\end{thm}

\begin{proof}
Whether $d_{1,vm} \neq g$ or not, to minimize the number of edge weight changes, we need to target the VNs that are connected to the maximum number of unsatisfied CNs. By definition, and since $b=d_1$, the number of VNs of this type is $a_{vm}$, and each is connected to $d_{1,vm}$ unsatisfied CNs.

In the case of $d_{1,vm} \neq g$, which is the general case, for any VN of the $a_{vm}$ pertinent VNs, there are $\binom{\gamma-d_{1,vm}}{E_{mu}}$ different ways of selecting $E_{mu}$ degree-$2$ satisfied CNs connected to this VN. Each of these CNs has $2$ edges we can change their weights (not simultaneously). Moreover, each edge can have $(q-2)$ different new weights (excluding the $0$ and the current weight). Thus, the number of candidate sets is:
\begin{equation}\label{eq_prch1}
S_{mu} = a_{vm}\binom{\gamma-d_{1,vm}}{E_{mu}}2^{E_{mu}}(q-2)^{E_{mu}},
\end{equation}
which is a rephrased version of (\ref{eq_choices1}).

In the case of $d_{1,vm} = g$, from (\ref{eq_emu}), $E_{mu}=1$ (the GAST is removed by a single edge weight change). Moreover,
\begin{equation}\label{eq_prch2}
\gamma-d_{1,vm} = \gamma-g = \left \lceil \frac{\gamma+1}{2} \right \rceil.
\end{equation}
Substituting (\ref{eq_prch2}) and $E_{mu}=1$ into (\ref{eq_prch1}) gives that the number of candidate sets follows the inequality:
\begin{equation}\label{eq_prch3}
S_{mu} \leq a_{vm} \left \lceil \frac{\gamma+1}{2} \right \rceil 2(q-2).
\end{equation}
In (\ref{eq_prch3}), the equality is achieved only if there are no shared degree-$2$ CNs between the VNs that have $g$ unsatisfied CNs, i.e., $n_{co}=0$. Otherwise, $n_{co}$ has to be subtracted from $a_{vm} \left \lceil \frac{\gamma+1}{2} \right \rceil$, which proves (\ref{eq_choices2}).

Note that the subtraction of $n_{co}$ is not needed if $d_{1,vm} \neq g$. The reason is that if $d_{1,vm} \neq g$ (or $E_{mu} \neq 1$) multiple edges connected to the same CN cannot exist in the same candidate set. Additionally, since our codes have girth at least $6$, there does not exist more than one degree-$2$ CN connecting the same two VNs in a GAST. 
\end{proof}

\begin{figure}[H]
\vspace{-2.5em}
\center
\includegraphics[width=3.5in]{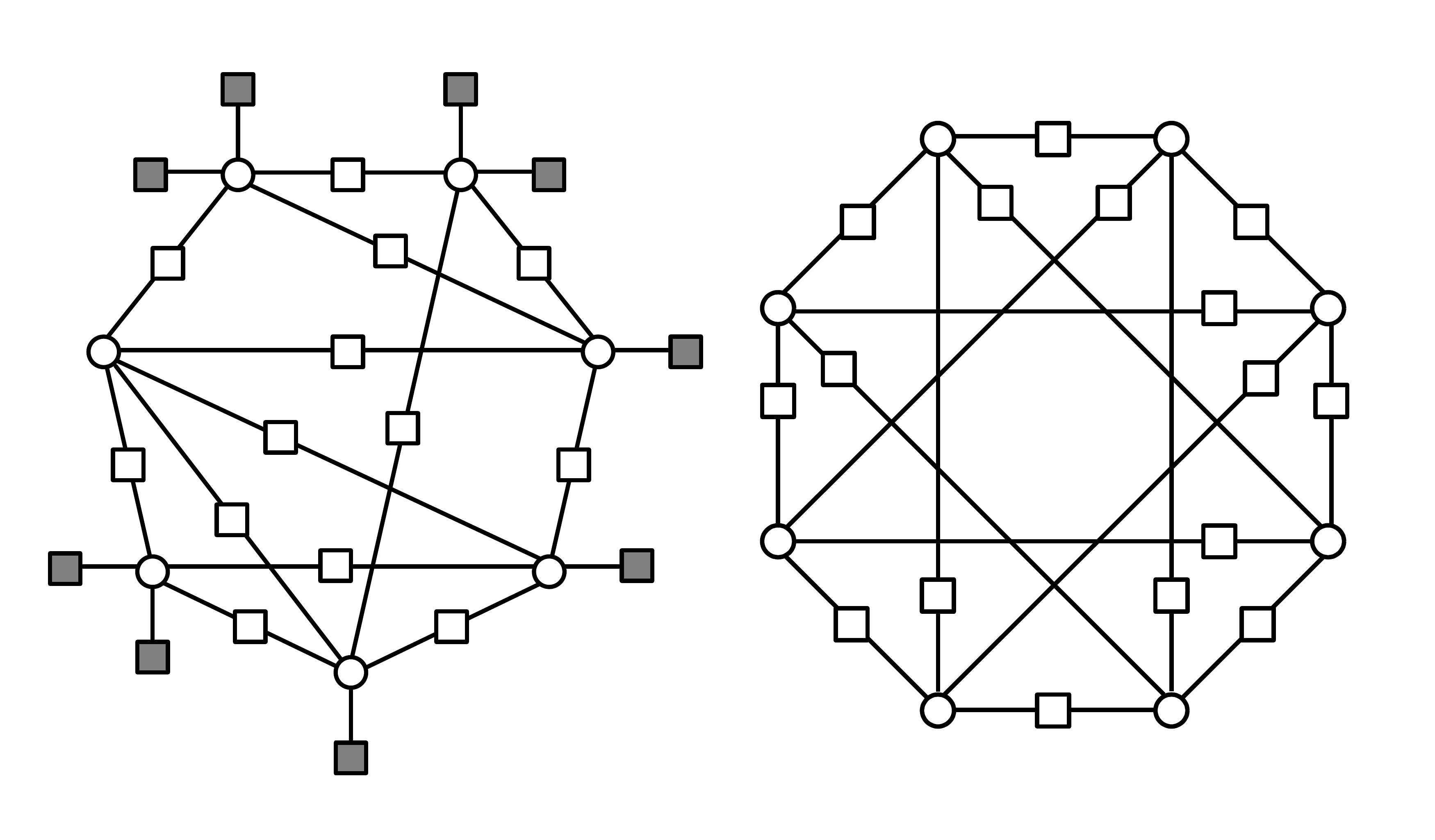}\vspace{-0.8em}
\vspace{-0.5em}
\text{\hspace{-0em}\footnotesize{(a) \hspace{14em} (b)}}
\caption{(a) A $(7, 9, 9, 13, 0)$ GAST ($\gamma=5$). (b) An $(8, 0, 0, 16, 0)$ GAST ($\gamma=4$). Appropriate non-binary edge weights are assumed.}
\label{Fig_choices}
\vspace{-0.6em}
\end{figure}

\begin{example}
Consider the $(7, 9, 9, 13, 0)$ GAST over GF($q$) in Fig. \ref{Fig_choices}(a) ($\gamma = 5$). For this GAST, $g=2$, $d_{1,vm}=2$, $a_{vm}=3$, and from (\ref{eq_emu}), $E_{GAST,min}=E_{mu}=1$. Moreover, $n_{co}=1$ (only one shared degree-$2$ CN between two VNs having two unsatisfied CNs each). Thus, from (\ref{eq_choices2}), the number of candidate sets of cardinality $1$ is:
\begin{equation}
S_{mu}=(3(3)-1)(2)(q-2)=16(q-2).
\end{equation}

Contrarily, for the $(8, 0, 0, 16, 0)$ GAST over GF($q$) in Fig. \ref{Fig_choices}(b) ($\gamma = 4$), $g=1$, $d_{1,vm}=0$, $a_{vm}=8$, and from (\ref{eq_emu}), $E_{GAST,min}=E_{mu}=2$. Thus, from (\ref{eq_choices1}) (the general relation), the number of candidate sets of cardinality $2$ is:
\begin{equation}
S_{mu}=8 \binom{4}{2} (2)^2 (q-2)^2 = 192(q-2)^2.
\end{equation}
\end{example}

\section{Code Design Steps and Simulation Results}\label{sec_sims}

In this section, we present our $\gamma = 3$ NB-SC code design approach for Flash memories, and the experimental results demonstrating its effectiveness. The steps of our OO-CPO-WCM approach are:

\begin{enumerate}
\item Specify the code parameters, $\kappa$, $p$, and $L$, with $m=1$.
\item Solve the optimization problem in (\ref{eq_optp}) for an optimal vector of overlap parameters, $\bold{t}^*$.
\item Using $\bold{H}^{bp}$ and $\bold{t}^*$, apply the circulant power optimizer to reach the powers of the circulants in $\bold{H}^{b}$ and $\bold{H}^{b}_{SC}$. Now, the binary image, $\bold{H}^{b}_{SC}$, is designed.
\item Assign the edge weights in $\bold{H}^{b}$ to generate $\bold{H}$. Next, partition $\bold{H}$ using $\bold{t}^*$, and couple the components to construct $\bold{H}_{SC}$.
\item Using initial simulations over a practical Flash channel and combinatorial techniques, determine the set $\mathcal{G}$ of GASTs to be removed from the graph of $\bold{H}_{SC}$.
\item Use the WCM framework (see \cite[Algorithm 2]{ahh_jsac}) to remove as many as possible of the GASTs in $\mathcal{G}$.
\end{enumerate}

In this section, the CV and MO results proposed are the best that can be achieved by these two techniques \cite{homa_sc, homa_mo}.

\vspace{-0.8em}
\begin{table}[H]
\caption{Number of $(3, 3, 3, 0)$ UGASTs in SC codes with $\gamma = 3$, $m = 1$, and $L = 30$ designed using different techniques.
}
\vspace{-0.5em}
\centering
\scalebox{0.85}
{
\begin{tabular}{|c|c|c|c|c|}
\hline

\multirow{2}{*}{Design technique} & \multicolumn{4}{|c|}{\makecell{Number of $(3, 3, 3, 0)$ UGASTs}} \\
\cline{2-5}
{} & $\kappa=p=7$ & \makecell{$\kappa=p=11$} & \makecell{$\kappa=p=13$} & \makecell{$\kappa=p=17$} \\
\hline
Uncoupled with AB & 8820 & 36300 & 60840 & 138720 \\
\hline
SC CV with AB & 3290 & 14872 & 25233 & 59024 \\
\hline
SC MO with AB & 609 & 3850 & 6851 & 15997 \\
\hline
SC best with AB & 609 & 3520 & & \\
\hline
SC OO-CPO with CB & 203 & 2596 & 5356 & 14960 \\
\hline
\end{tabular}}
\label{table_1}
\end{table}
\vspace{-0.5em}

We start our experimental results with a table comparing the number of $(3, 3, 3, 0)$ UGASTs in SC codes designed using various techniques. All the SC codes have $\gamma=3$, $m=1$, and $L=30$. AB codes are used as the underlying block codes in all the SC code design techniques we are comparing the proposed OO-CPO technique against. Table \ref{table_1} demonstrates reductions in the number of $(3, 3, 3, 0)$ UGASTs achieved by the OO-CPO technique over the MO technique (resp., the CV technique) that ranges between $6.5\%$ and $66.7\%$ (resp., $74.7\%$ and $93.8\%$). More intriguingly, the table shows that the OO-CPO technique provides lower number of $(3, 3, 3, 0)$ UGASTs than the best that can be achieved if AB underlying block codes are used. Note that this ``best'' is reached using exhaustive search, and that is the reason why we could not provide its counts for $\kappa = p > 11$.

Next, we provide simulation results verifying the performance gains achieved by our NB-SC code design approach for Flash memories. The Flash channel we use is a practical Flash channel, which is the normal-Laplace mixture (NLM) Flash channel \cite{mit_nl}. Here, we use $3$ reads, and the sector size is $512$ bytes. We define RBER as the raw bit error rate \cite{ahh_jsac}, and UBER as the uncorrectable bit error rate \cite{ahh_jsac}. One formulation of UBER, which is recommended by industry, is the frame error rate (FER) divided by the sector size in bits. Simulations were done in software on a high speed cluster of machines.

All the NB-SC codes we simulated are defined over GF($4$), and have $\gamma = 3$, $\kappa = p =19$, $m = 1$, and $L = 20$ (block length $=14440$ bits and rate $\approx 0.834$). Code 1 is uncoupled (AB). Code 2 is designed using the CV technique. Code 3 is designed using the OO technique (with no CPO applied). The underlying block codes of Codes 2 and 3 are AB codes. Code 4 is designed using the OO-CPO technique. The edge weights of Codes 1, 2, 3, and 4 are selected randomly. Code 5 (resp., Code 6) is the result of applying the WCM framework to Code 1 (resp., Code 4) to optimize the edge weights.

Code 1 (resp., Code 2 and Code 4) has $129960$ (resp., $55366$ and $16340$) $(3, 3, 3, 0)$ UGASTs. Additionally, Code 1 (resp., Code 2 and Code 4) has $4873500$ (resp., $2002353$ and $1156264$) $(4, 4, 4, 0)$ UGASTs. The $(4, 4, 4, 0)$ UGAST is the second most common substructure in the dominant GASTs of NB codes with $\gamma=3$ simulated over Flash channels.

\begin{figure}[H]
\vspace{-1.0em}
\center
\includegraphics[trim={0.4in 0.2in 0.5in 0.2in},clip,width=3.5in]{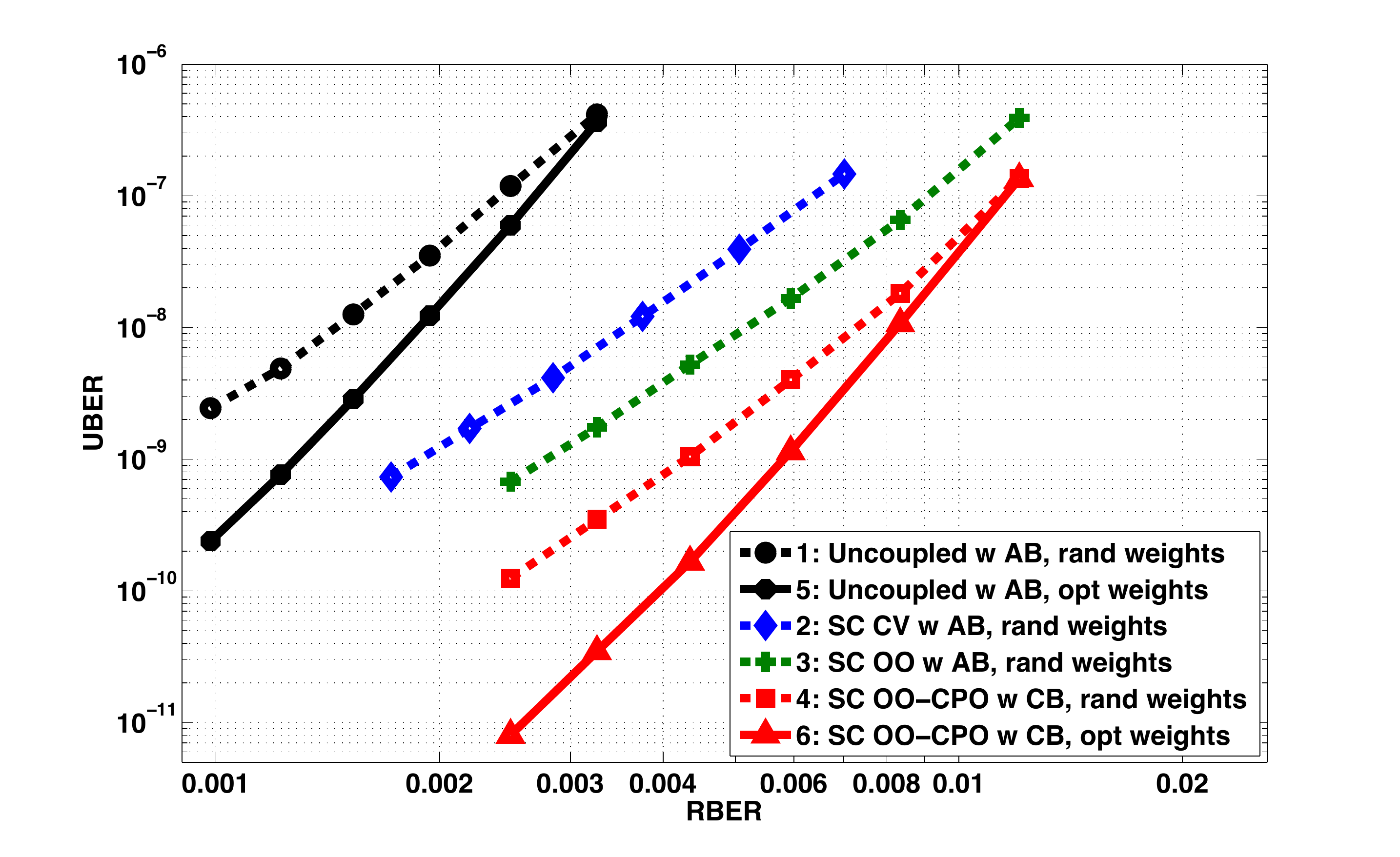}
\vspace{-1.7em}
\caption{Simulation results over the NLM Flash channel for SC codes with $\gamma = 3$, $m = 1$, and $L = 20$ designed using different techniques.}
\label{Fig_uber}
\vspace{-0.5em}
\end{figure}

Fig. \ref{Fig_uber} demonstrates the performance gains achieved by each stage of our NB-SC code design approach. Code 3 outperforms Code 2 by about $0.6$ of an order of magnitude, which is the gain of the first stage (OO). Code 4 outperforms Code 3 by about $0.7$ of an order of magnitude, which is the gain of the second stage (CPO). Code 6 outperforms Code 4 by about $1.2$ orders of magnitude, which is the gain of the third stage (WCM). Moreover, the figure shows that the NB-SC code designed using our OO-CPO-WCM approach, which is Code 6, achieves about $200\%$ (resp., more than $500\%$) RBER gain compared to Code 2 (resp., Code 1) over a practical Flash channel. An intriguing observation we have encountered while performing these simulations is the change in the error floor properties when we go from Code 2 to Code 4. In particular, while the $(6, 0, 0, 9, 0)$ GAST was a dominant object in the case of Code 2, we have encountered very few $(6, 0, 0, 9, 0)$ GASTs in the error profile of Code 4.

\section{Conclusion}\label{sec_conc}

We proposed a combinatorial approach for the design of NB-SC codes optimized for practical Flash channels. The OO-CPO technique efficiently optimizes the underlying topology of the NB-SC code, then the WCM framework optimizes the edge weights. NB-SC codes designed using our approach have reduced number of detrimental GASTs, thus outperforming existing NB-SC codes over Flash channels. The proposed approach can help increase the reliability of ultra dense storage devices, e.g., emerging 3-D Flash devices.

\section*{Acknowledgement}\label{sec_ack}

The research was supported in part by a grant from ASTC-IDEMA and by NSF.


\end{document}